\documentclass{llncs}
\usepackage[width=440pt,height=650pt,centering]{geometry}
\usepackage[english]{babel}
\usepackage{xspace}
\usepackage{framed}
\usepackage{amsmath}
\usepackage{amssymb}
\usepackage{graphicx}
\usepackage{algorithm2e}
\usepackage{verbatim}
\usepackage{multirow}
\usepackage{calc}
\usepackage{enumerate}
\usepackage{color}

\usepackage[OT2,OT1]{fontenc}
\newcommand{\YCorr}[1]{{{#1}}\xspace}

\DeclareMathAlphabet{\mathsc}{OT1}{cmr}{m}{sc}

\newcommand{\BigO}[1]{\mathcal{O}(#1)}
\newcommand{\BBigO}[1]{\mathcal{O}\left(#1\right)}
\newcommand{\Prob}[1]{\mathbb{P}_{x,\Weights}(#1)}
\newcommand{\WProb}[1]{\mathbb{P}(#1)}

\newcommand{\Comp}[1]{{\boldsymbol \CompF}(#1)}
\newcommand{\CompE}[2]{{\boldsymbol \CompF}(#1)_{#2}}
\newcommand{\CompF}{p}
\newcommand{\Target}{f}
\newcommand{\TargComp}{{\boldsymbol \Target}}

\newcommand{\ErrFactor}{{\boldsymbol{\ErrFactorEntry}}}
\newcommand{\ErrFactorEntry}{\epsilon}
\newcommand{\ErrExp}{{{\alpha}}}

\newcommand{\RatLang}{{\mathcal{L}}}
\newcommand{\RatGF}{{L_{\Weights}}}

\newcommand{\Struct}{w}
\newcommand{\At}{t}
\newcommand{\Empty}{1}
\newcommand{\WF}{\pi}
\newcommand{\Weights}{\boldsymbol{\WF}}
\newcommand{\AtW}[1]{\WF_{#1}}
\newcommand{\Weight}[1]{\WF(#1)}
\newcommand{\Atoms}{\Sigma}

\newcommand{\Olive}[1]{{#1}}
\newcommand{\CS}[1]{\mathcal{#1}}
\newcommand{\Gen}[1]{{\displaystyle{\mathrm{\Gamma}}} #1_{\WF}(x)}

\newcommand{\GF}[1]{#1}
\newcommand{\WGFun}[1]{\GF{#1}_{\Weights}}
\newcommand{\WGF}[1]{\WGFun{#1}(z)}
\newcommand{\WGFb}[2]{\WGFun{#1}(#2)}
\newcommand{\PFSymb}{c}
\newcommand{\PF}[1]{\PFB{\WF}{#1}}
\newcommand{\PFB}[2]{\PFSymb_{#1,#2}}
\newcommand{\Freqs}{\boldsymbol{\mu}}

\newcommand{\Def}[1]{\emph{#1}}
\DeclareMathOperator*{\Bern}{\mbox{Bern}}

\DeclareMathOperator*{\Seq}{\mbox{\sc Seq}}

\begin{document}

\title{Multi-dimensional Boltzmann Sampling of Languages}%
\author{Olivier Bodini\inst{1} \and Yann Ponty\inst{2}}%
\institute{%
{Laboratoire d'Informatique de Paris 6 (LIP6), CNRS UMR 7606\\
Universit\'e Paris 6 - UPMC, 75252 Paris Cedex 05, France}
\and {Laboratoire d'Informatique de l'\'ecole Polytechnique (LIX), CNRS UMR 7161/AMIB INRIA\\
\'Ecole Polytechnique, 91128 Palaiseau, France}}

\maketitle
\begin{abstract}
This paper addresses the uniform random generation of  words from a context-free language (over an alphabet of size $k$),
while constraining every letter to a targeted frequency of occurrence. Our approach consists in a multidimensional extension of Boltzmann samplers~\cite{Duchon2004}.
We show that, under mostly \emph{strong-connectivity} hypotheses, our samplers return a word of size in $[(1-\varepsilon)n, (1+\varepsilon)n]$ and exact frequency in $\mathcal{O}(n^{1+k/2})$ expected time.


Moreover, if we accept tolerance intervals of width in $\Omega(\sqrt{n})$ for the number of occurrences of each letters, our samplers perform an approximate-size generation of words in expected $\mathcal{O}(n)$ time. We illustrate these techniques on the generation of Tetris tessellations with uniform statistics in the different types of tetraminoes.
\end{abstract}
\section{Introduction}
 Random generation is the core of the simulation of complex data. It
appears in real applicative domains such as complex networks (biology, Internet or social relationship), or software testing (validation, benchmarking).
It helps us to predict the behavior of algorithms (complexities and statistical significance of results), to visualize limit properties (such as transition
phases in statistical physics), to model real contexts (random graphs for web simulation).

  Following the pioneering work of Flajolet \emph{et al}~\cite{Flajolet1994}, decomposable combinatorial classes can be specified
  using standard specifications. Two major techniques can then be applied to draw $m$ objects of size $n$ at random from such a class.
  On one hand, the recursive approach~\cite{Wilf1977} precomputes the cardinalities of sub-classes for sizes up to $n$ and uses these numbers to
  perform local choices that are consistent with the targeted uniformity. The best known optimization of this technique~\cite{Denise2000} uses
  certified floating point arithmetics and works in $\BigO{m\cdot n^{1+o(1)}}$ but its implementation remains highly non-trivial due to its sophisticated precomputations.
  On the other hand, the Boltzmann sampling techniques, recently introduced by Duchon \emph{et al}~\cite{Duchon2004}, achieves a random generation
  for most unlabelled~\cite{Flajolet2007} and labelled specifications in $\BigO{m\cdot n^2}$ operations at an optimally low  $\BigO{m\cdot n}$ memory cost.
  Instead of enforcing a strict -- and costly -- control on the size of generated objects, this general technique rather induces an appropriate
  distribution on the size of sampled objects, and performs rejection until a suitable object is found.

  In the present work, we investigate a natural multivariate extension of Boltzmann sampling, aiming at drawing objects, uniformly at random,
  having a prescribed composition in the different terminal letters. From a combinatorial perspective, such a generation allows the so-called
  symbolic method to reclaim combinatorial classes and languages that \emph{fall slightly off} of its natural expressivity. For instance, restrictions
  of rational languages may not admit a rational (or even context-free) specification under the additional hypothesis that some letters co-occur strictly
  (One may consider the triple-copy language).
  For context-free languages on $k$ letters, this problem was previously addressed within the recursive framework~\cite{Wilf1977} by Denise \emph{et al}~\cite{Denise2000}, deriving algorithms in
  $\Theta(n^k)$ and $\Theta(n^{2k})$ arithmetic operations, respectively for rational and context-free languages.
  Using properties of holonomic series, Bertoni\emph{ et al}~\cite{Bertoni2003} revisited the problem and proposed a method for the uniform sampling from rational languages on two letters in $\Theta(n)$.
  Unfortunately a direct generalization of the technique yields an algorithm in $\Theta(n^{k-1})$ for $k$ letters, as pointed out in Radicioni's thesis~\cite{Radicioni2006}.

  Following the general philosophy of Boltzmann sampling, our algorithm will first relax the compositional constraint, using non-uniform samplers to draw objects whose average
  composition is fine-tuned to match the targeted one, and perform rejection until an acceptable object is found. By acceptable, one understands that
  generated objects must feature prescribed size and composition, while tolerances may be allowed on both requirements.
  Our programme can then be summarized in the three following phases:
  \begin{enumerate}[{Phase }I.]
    \item\label{phase:ComputeWeights} Figure out a set of weights such that the expected composition matches the targeted one.
    \item\label{phase:GenerateWeighted} Draw structures from a weighted distribution, using either the recursive approach (See~\cite{Denise2000})
    or a weighted Boltzmann sampler (See Section~\ref{sec:WeightedBolztmannSamplers}).
    \item\label{phase:RejectComposition} Reject structures of unsuitable compositions, until an adequate object is generated and returned.
  \end{enumerate}
  Although phases \ref{phase:GenerateWeighted} and \ref{phase:RejectComposition} are independently addressed in our analyses,
  one can (and will) combine them into a single rejection step when a weighted Boltzmann sampler is used for Phase~\ref{phase:GenerateWeighted}.
  The algorithmic aspects of our programme will essentially build on and extend previous works addressing the uniform version, but a general analysis
  of its overall performance is more challenging. Indeed, the complexity of the rejection Phase~\ref{phase:RejectComposition} is heavily related to a general
  analysis of the limiting distribution of the associated multivariate -- parameter-induced -- generating functions.
  For each phase, we attempt to give mathematical characterizations of classes having proper behaviors. In particular, for context free languages whose grammars are \Def{strongly connected} and \Def{aperiodic},
  we obtain for each combination of tolerances, the complexities summarized in Table~\ref{tab:compSummary}.
  \begin{table}[t]
  \centering
  {\setlength{\tabcolsep}{8pt}
  \renewcommand{\arraystretch}{1.5}
  \YCorr{\begin{tabular}{| c | c | c | c |}\hline
     \multicolumn{2}{|c|}{\multirow{2}{*}{{Tolerance}}} & \multicolumn{2}{|c|}{Composition}\\ \cline{3-4}
     \multicolumn{2}{|c|}{} & None  &  $\Omega(\sqrt{n})$\\ \hline
     \multirow{2}{*}{{Size}} & None  &  $\BigO{n^{2+k/2}}$ & $\BigO{n^{2}}$ \\
       &  $\Theta({n})$  & $\BigO{n^{1+k/2}}$ & $\BigO{n}$\\ \hline
  \end{tabular}}\\[1em]
  }
  \caption{\YCorr{Average-case complexities of our samplers for a word of length $n$ over $k$ letters in strongly connected context-free languages under different tolerances.}}\label{tab:compSummary}
  \end{table}

The plan of this paper follows the different phases :  Section 2 defines the concepts and notations used throughout the
paper. Section~3 explains how to tune efficiently the parameters such that the targeted composition matches the average behavior (Phase~\ref{phase:ComputeWeights}).
In Section~4, we discuss the complexity of Phase~\ref{phase:GenerateWeighted}, the number of rejections needed to reach a word of suitable size
(or suitable approximate size). The complexity of the multidimensional rejection  (Phase~\ref{phase:RejectComposition}) is addressed in Section~5.
We illustrate our method in Section~6 by sampling perfect Tetris tessellations -- tessellations of a $w\times h$ rectangles using balanced lists of tetraminoes.
Finally we conclude with a short overview of future works.

\section{Notations and definitions}



%
  Following traditional mathematical notations, we will use bold symbols for multi-dimensional variables/functions (i.e. ${\boldsymbol x}$), and use  subscripts to access a specific dimension (i.e. $x_i$).
  Throughout the rest of the document, we will denote by $\Atoms$ the \Def{alphabet} of $k$ letters, by $\CS{C}$ a \Def{context-free language} over $\Atoms$, and by $n$
  the length of generated words.

  \subsubsection*{Composition and tolerance.}
    Define the  \Def{composition} of sampled words as the \Def{frequency} of occurrences of each letter $\At_i$ in a
    word $\Struct\in\CS{C}$, denoted by $\Comp{\Struct} := \left(|\Struct|_{\At_i}/n\right)_{i\in [1,k]}.$
    Our main goal is to generate -- uniformly at random -- some word $\Struct\in\CS{C}$ having a composition that is \emph{close to} a \Def{targeted
    composition} $\TargComp\in[0,1]^k$ such that $\sum_{i\in[1,k]} \TargComp_i = 1$.

    We make this notion of proximity explicit, and formalize the notion of acceptability for a sampled word.
    Namely let $\ErrFactor$ be a $k$-tuple of positive real numbers and $\ErrExp\in \mathbb{Q}^+$ a rational exponent, an object $\Struct\in\CS{C}$
    qualifies as \Def{$(\ErrFactor, \ErrExp)$-acceptable} if and only if
    $$ \CompE{\Struct}{i} \in I(\Target_i, \ErrFactorEntry_i,\ErrExp) ,\;\text{for all } i\in[1,k]$$
    where \Olive{$I(f,e,a) := [f-f^{a}n^{a-1} e, f+f^{a}n^{a-1} e]$.}
    This definition captures the case of fixed (exact) compositions by setting $\ErrExp=1$ and $\ErrFactorEntry_i = 1/n, \text{for all } i\in[1,k]$.

  \subsubsection*{Weighted distributions.}
   \begin{figure}[t]
      \begin{equation*}
      \renewcommand{\arraystretch}{1.8}
      \setlength{\tabcolsep}{8pt}
         \YCorr{\begin{array}{| l l l l|}\hline \label{tab1}
          \text{ Epsilon} & \CS{C} = 1 &  \WGF{C} = 1 & \Gen{C} := \varepsilon\\
          \text{ Letters} & \CS{C} = \At_i &  \WGF{C} = \AtW{\At_i}z & \Gen{C} := \At_i\\
          \text{ Union} & \CS{C} = \CS{A} + \CS{B} & \WGF{C} = \WGF{A} +\WGF{B} & \Gen{C} := \Bern\left(\displaystyle\frac{\WGFb{A}{x}}{\WGFb{C}{x}},\frac{\WGFb{B}{x}}{\WGFb{C}{x}}\right) \longrightarrow \Gen{A}\;|\;\Gen{B}\\
          \text{ Product} & \CS{C} = \CS{A} \times \CS{B} & \WGF{C} = \WGF{A}\times \WGF{B} & \Gen{C} := \Gen{A}.\Gen{B}\\ \hline
        \end{array}}
      \end{equation*}
      \caption{ \YCorr{Weighted generating functions and associated Boltzmann sampler $\Gen{C}$ for context-free languages.}\label{fig:constructions}}
    \end{figure}

    The following notions and definitions, recalled here for the sake of self-containment, can be found in Denise~\emph{et al}~\cite{Denise2000}.
    A positive \Def{weight} vector $\Weights$ assigns positive weights $\WF_i\in\mathbb{R}^+$ to each letter $\At_i\in\Atoms$.
    The weight is then extended multiplicatively on any object $\Struct$ by $\Weight{\Struct} = \prod_{x \in \Struct} \AtW{x}.$ 
    This gives rise to the notion of \Def{weighted generating function} $\WGF{C}$ for a context-free language $\CS{C}$, a natural
    generalization of the size (enumerative) generating function where each structure is counted with multiplicity equal to its weight
    $$ \WGF{C} = \sum_{\Struct\in\CS{C}} \Weight{\Struct}z^{|\Struct|}
    \YCorr{= \sum_{\Struct\in\CS{C}} \WF_{\At_{1}}^{|\Struct|_{\At_{1}}}\cdots \WF_{\At_k}^{|\Struct|_{\At_k}} z^{|\Struct|}}
    = \sum_{n\ge 0} \PF{n} z^n$$
    where $\PF{n}$ is the total weight\footnote{This quantity is essentially similar to the partition function in statistical mechanics, introduced by L. Boltzmann.}
    of objects of size $n$. Notice that this generating function can be re-interpreted as a multivariate generating function in $\Weights$ and $z$

    This weighting scheme implicitly defines a \Def{weighted distribution} on the set $\CS{C}_n$ of words of size $n$, such that
    $$ \WProb{\Struct\;|\;\YCorr{n=|\Struct|}} = \frac{\Weight{\Struct}}{\sum_{\Struct'\in \CS{C}_n} \Weight{\Struct'}} = \frac{\Weight{\Struct}}{\PF{n}}.$$

    Finally, the weighted distribution \YCorr{generalizes to} a {\bf Boltzmann weighted distribution} on the whole language such that
    \begin{equation}\Prob{\Struct\;|\;\YCorr{n=|\Struct|}} = \frac{\Weight{\Struct}x^n }{\sum_{\Struct'\in \CS{C}} \Weight{\Struct'}x^{|\Struct'|} } = \frac{\Weight{\Struct}x^n}{\WGFun{C}(x)}.\label{eq:WeightedBoltzmannDistribution}\end{equation}
\begin{property}
Let $N$ (resp. $N_i$) be the random variable \YCorr{associated with} the size (resp. number of occurrences of a letter $\At_i$) of a word in
a $(x,\Weights)-$ Boltzmann weighted distribution over a class $\CS{C}$.
Then the \YCorr{expectations of $N$ and $N_i$}
are \YCorr{related to} the partial derivatives of the multivariate generating function $\WGF{C}$ through
\begin{equation}
 \mathbb{E}_{x,\Weights}(N) =  x\dfrac{\frac{d\WGFun{C}(x)}{dx}}{\WGFun{C}(x)}  \text{\YCorr{\quad \quad and \quad \quad  }} \mathbb{E}_{x,\Weights}(N_i)  =  \dfrac{\WF_i\frac{\partial}{\partial \WF_i}\WGFun{C}(x)}{\WGFun{C}(x)}\label{eq:expectations}
\end{equation}
\end{property}
In the sequel we will denote by $\Freqs(x,\Weights)$ 
the vector of expectations $(\mathbb{E}_{x,\Weights}(N_1),\cdots,\mathbb{E}_{x,\Weights}(N_k))$.

 \section{Tuning weights (Phase~\ref{phase:ComputeWeights}) }\label{tuning}

First, let us  address the question of finding a vector $\Weights$ such that the
multidimensional rejection scheme (Phase~\ref{phase:RejectComposition}) is as efficient as possible. We propose and explore two
alternatives, both computing a weights vector that make the expected and targeted compositions coincide.
The first one uses a numerical Newton iteration. The second one uses an asymptotic approximation for
the value of $z$ which greatly simplifies the weights/frequencies relationship.


\paragraph{Tuning by expectation.} Newton's methods are based on successive linear (or higher order) approximations in order to obtain numerical estimates of a root of a system of equations. It is generally an efficient algorithm assuming that the initial values are close enough to a root. Here, we are interested in finding the unique root $(z_0,\Weights_{\TargComp})$ of the system  $\Freqs(z_0,\Weights)=n\TargComp$. Algorithm~\ref{algo1} is a slightly revisited version of Newton's method which tests at each step if Newton's approximation has improved the estimate of the root. This test fails if and only if the current parameters are too far from the solution. In this case,
we search using dichotomy an intermediate target that is closer to the solution than the current parameters.

\begin{proposition}\label{prop0}
Let $\TargComp$ and $n$ be the targeted composition and size respectively.
Assume that the Jacobian matrix $J(\mathbb{E}_{z_0}(\Weights_{\TargComp}))$ is not singular\footnote{I.e. there is no linear dependency between the expected numbers of different letters.}, then Algorithm~\ref{algo1} returns $(z_0,\Weights_1)$ such that the expected composition $\Freqs(z_0,\Weights_1)$ satisfies $||\Freqs(z_0,\Weights_1)-n\TargComp||<\epsilon$.\\
 Moreover, there exists a neighborhood $B$ of $(z_0,\Weights_{\TargComp})$ such that, for any $\Weights_0\in B$, Algorithm~\ref{algo1} with
 initial condition $\Weights_0$ quadratically converges to $\Weights_{\TargComp}$ (i.e. $\exists C>1$ such that $\forall k\geq 0, ||\Weights_k-\Weights_{\TargComp}||\leq C^{-2k}$ where $\Weights_{k+1} := J(\mathbb{E}_{z_0})^{-1}(\Weights_k)\cdot(n\TargComp-\mathbb{E}_{z_0}(\Weights_k))+\Weights_k$).
\end{proposition}

 \begin{algorithm}[t]\label{algo1}\begin{framed}
\caption{Tracking the weights.}

\KwIn{Initial parameters $z_0$ and $\Weights_0$, a composition $\TargComp$, a size $n$ and $\epsilon$ a numerical precision}

\KwOut{The valid weights}
Let $\mathbb{E}_{z_0}$ be the map  from the space of the weights into $\mathbb{R}^k_+$ such that $\mathbb{E}_{z_0}(\Weights)=\Freqs(z_0,\Weights)$\;
Let $J(\mathbb{E}_{z_0}(\Weights))$ be the Jacobian matrix of \YCorr{$\mathbb{E}_{z_0}$ evaluated at $\Weights$}\;
$\Weights:=\Weights_0$\;
\Repeat{end=true}{
end:=true; $\boldsymbol{c}:=n\TargComp$; $N:=||\boldsymbol{c}-\mathbb{E}_{z_0}(\Weights)||$\;
\While{$N>\epsilon$}{
$\Weights_{aux} := \Weights$\;
$\Weights := J(\mathbb{E}_{z_0})^{-1}(\Weights)\cdot(n\TargComp-\mathbb{E}_{z_0}(\Weights))+\Weights$\;
\If{$N<||\boldsymbol{c}-\mathbb{E}_{z_0}(\Weights)||$}{$\Weights:=\Weights_{aux}$; $\boldsymbol{c}:=(\boldsymbol{c}+\mathbb{E}_{z_0}(\Weights))/2$; end:=false;}
}}
\Return $\Weights$
\end{framed}\end{algorithm}


\paragraph{Asymptotic tuning.}
Since one generally attempts to generate large objects, a natural option consists in solving the simpler asymptotic system.

\begin{proposition}\label{prop0a}
Let us consider a language whose grammar is irreducible and aperiodic and whose generating function $\WGF{C}$ admits $\rho(\Weights)$ as dominant
singularity. \Olive{Then\YCorr{, for any letter $\At$ and as $z$ tends to $\rho(\Weights)$,} it holds that:}

{\centering\YCorr{\begin{tabular}{cl}
$\mathbb{E}_{z,\Weights}(N_\At)\sim \frac1{2}\WF_\At n\frac{\frac{\partial}{\partial \WF_\At}\rho(\Weights)}{\rho}$ & if $\rho(\Weights)$ is a rational singularity,\\
$\mathbb{E}_{z,\Weights}(N_\At)\sim -\WF_\At n\frac{\frac{\partial}{\partial \WF_\At}\rho(\Weights)}{\rho}$ & if $\rho(\Weights)$ is an algebraic singularity.
\end{tabular}}\\}
\end{proposition}


\begin{remark} Considering the expectation $\mathbb{E}_n(N_\At)$ of the number of letters $\At$ in a word of \Def{fixed} size $n$. Then, from~\cite{Denise2000}, similar asymptotic estimates 
 holds for $\mathbb{E}_n(N_\At)$ and the weights computed by our methods can therefore be used by the recursive approach.
\end{remark}

\section{Efficiency of the size rejection scheme (Phase~\ref{phase:GenerateWeighted})}

At this point, we assume that a $k$-tuple of weights $\Weights$ has been found such that the average composition in the
weighted distribution matches the targeted one. We now need to perform a random generation of $m$ words
from the context-free language with respect to the $\Weights$-weighted distribution.

This problem was previously addressed in \cite{Denise2000} within the framework of
the recursive method, and an algorithm in $\BigO{m\cdot n}$ arithmetic operations was proposed.
Despite its apparent low complexity, the exponential growth of the numbers processed by the algorithm
increases the practical complexity to $\Theta(m\cdot n^2)$ in time and $\Theta(n^2)$ in memory.

\label{sec:WeightedBolztmannSamplers}
Let us investigate a weighted generalization of Boltzmann sampling.
First let us remind that Boltzmann sampling first relaxes the size constraint and  draws objects in a Boltzmann distribution of parameter $x$.
To that purpose a fixed stochastic process, coupled with an (anticipated) rejection procedure,
 is used (See \Olive{Algorithm~\ref{algo:rejection}}). The probabilities of the different alternatives are precomputed by an external procedure
  called \Def{oracle} (Symbolic algebra, or numerical method in \cite{Pivoteau2008}).
  A judicious choice of value for $x$ ensures a low probability of rejection and this approach yields, for large classes of structures (trees, sequences, runs,
  mappings, fountains\ldots), generic algorithms in $\BigO{n^2}$ for objects of \Def{exact-size} $n$,
  and in $\BigO{n}$ for objects of \Def{approximate-sizes} in $[n(1-\varepsilon),n(1+\varepsilon)]$, for some $\varepsilon>0$.

  Through a minor modification of the oracle, one can easily turn unlabelled Boltzmann samplers, introduced in \cite{Flajolet2007}, into generators
  for the weighted Boltzmann distribution (See Equation~\ref{eq:WeightedBoltzmannDistribution}).
  Namely, one only needs to replace any occurrence of the generating function $\GF{C}(z)$ \YCorr{by} its weighted counterpart $\WGF{C}$, obtaining generic samplers
  summarized in Figure~\Olive{1}, and use the classic size rejection process (Algorithm~\ref{algo:rejection}).

\begin{algorithm}[t]
\begin{framed}
\caption{Rejection algorithm $\Gamma_2\mathcal{A}{(x,\Weights;n,\varepsilon)}$ \label{algo:rejection}}

\KwIn{Parameters $x,\Weights$}

\KwOut{Object of $\mathcal{A}$ of size in $I(n, \varepsilon) := [n(1-\varepsilon), n(1+\varepsilon)]$}
\Repeat{$|\gamma| \in I(n, \varepsilon)$ }{$\gamma := \Gen{A}$}
\Return($\gamma$)
\end{framed}
\end{algorithm}

  \begin{proposition}\label{prop1}
    Let $\Weights$ be a $k$-tuple of weights, $x$ be a Boltzmann parameter, $C$ be a context-free specification
    and $\WGF{C}$ its weighted generating function.\\
    Then \YCorr{the samplers $\Gen{C}$ summarized in Figure~\Olive{1} generate} any word $w\in\CS{C}$ with probability
    $$\Prob{\Struct\;|\;n} = \frac{\Weight{\Struct}x^n}{\WGFun{C}(x)}.$$
  \end{proposition}
    The (renormalized) restriction of a $\Weights$-weighted Boltzmann distribution to objects of size $n$ is clearly a
    $\Weights$-weighted distribution, and this fact ensures the correctness of a rejection-based approach.

Let us qualify a context-free language as \Def{well-conditioned} iff the singular exponent $\alpha_{\Weights}$ of its dominant singularity is non negative.
\YCorr{Following \cite{Duchon2004}, we observe that any grammar can be pointed repeatedly until the exponent of its generating function becomes non-negative. Moreover the pointing operator leaves a weighted distribution unaffected within the subset of words of a given length. Therefore we can restrict our analysis to grammars associated with \emph{flat} Boltzmann \Olive{distributions}, generate words from the pointed grammars and \emph{erase} the point(s) afterward.}

\begin{theorem}[Essentially proven in \cite{Duchon2004}]\label{theophase2}
Let $\CS{C}_{\Weights}$ be a weighted well-conditioned context-free language and $x_n$  be the root in $[0, \rho_{\Weights})$ of  $\mathbb{E}_{x,\Weights}(N)=n$.
Then the complexity $X_{\varepsilon}[n]$ of the sampler $\Gamma_2 \CS{C}(x_n,\Weights; n,\varepsilon)$ described in Algorithm~\ref{algo:rejection} is such that
\begin{itemize}
  \item If $\varepsilon=0$ (exact size): $X_{\varepsilon}[n] \in \BBigO{\dfrac{\kappa\Gamma(\alpha_{\Weights})n^2}{\alpha_{\Weights}^{\alpha_{\Weights}}}+c(\Weights)n}$, and
  \item If $\varepsilon>0$ (approximate-size): $X_{\varepsilon}[n] \in \BBigO{\dfrac{\kappa n}{\zeta_{\alpha_{\Weights}}(\varepsilon)}+c(\Weights)}$
\end{itemize}
where $\kappa$ is the cost-per-letter induced by the canonical Boltzmann samplers, $\alpha_{\Weights}$ is the singular exponent of the dominant singularity of $\WGF{C}$, $\zeta_{\alpha_{\Weights}}(\varepsilon):=\dfrac{\alpha_{\Weights}^{\alpha_{\Weights}}}{\Gamma(\alpha_{\Weights})}\displaystyle\int_{-\varepsilon}^{\varepsilon}(1+s)^{\alpha_{\Weights}-1}e^{-\alpha_{\Weights}(1+s)}ds$,
$\Gamma(x)$ is the gamma function, and $c(\Weights)$ does not depend on $n$.
\end{theorem}


In particular, for any fixed weight vector $\Weights$, Theorem~\ref{theophase2} implies a $\BigO{n}$ (resp. $\BigO{n^2}$) complexity for the approximate-size (resp. exact size) weighted
samplers. By contrast, using weights to enforce compositions that are unnatural (e.g. enforcing $\BigO{\sqrt{n}}$ occurrences of a letter occurring $\BigO{n}$ times in the uniform distribution) may lead
to a -- somewhat hidden --
dependency of $\Weights$ in $n$. Although we were unable to characterize these dependencies and their impact $c(\Weights)$ on both complexities, we expect the latter to remain limited, and conjecture similar complexities when \emph{meaningful} compositions are targeted. \YCorr{For instance, assuming at least one occurrence of each letter (a realistic assumption, since prohibition of a letter is simply achieved through a grammar modification), and the frequencies and the weights can therefore be assumed to be bounded by some function of $n$.}

In the case of rational languages, the following theorem provides a computable evaluation for the efficiency of the size-rejection process.  It relies on the partial fraction expansion of rational functions, which can be obtained for any weighted generating function $C_{\Weights}(z)$, and is denoted by
   \begin{equation}\label{eq1} C_{\Weights}(z)=\sum\limits_{i=1}^{r}\sum\limits_{k=1}^{m_i}(1-z/\rho_i)^{-\alpha_{i,k}}h_{i,k}+P(z)
    \end{equation}
  where \YCorr{$P(z)$ is a polynomial of degree bounded by a constant, $r$ the number of distinct singularities} and $m_i$ the multiplicity of $\rho_i$ which are sorted by increasing module.
  In weighted generating functions \YCorr{the values of} $\rho_i$, $P(z)$, $h_{i,k}$, $k$ and $r$ depend on the actual values of the weights.
\begin{theorem}
Let $\CS{C}_{\Weights}$ be a weighted rational language, $x_n$  be the root in $[0, \rho_{\Weights})$ of  $\mathbb{E}_{x,\Weights}(N)=n$ \YCorr{and $\varepsilon>0$ be a tolerance then the approximate-size sampler $\Gamma_2 \CS{C}(x_n,\Weights; n,\varepsilon)$ succeeds after an expected number of trials of $\Gamma C_{\Weights}(x,b)$ in}
$$\dfrac{C_{\Weights}(x_n)}{\left(\sum\limits_{i=1}^{r}\sum\limits_{k=1}^{m_i}\binom{n+k-1}{k-1}(\rho_i)^{-n}h_{i,k}+[z^n]P(z)\right)(x_n)^n}.$$
\end{theorem}

\section{Complexity of the multidimensional rejection (Phase \ref{phase:RejectComposition})}
    Our approach relies on a rejection scheme that generalizes that of the classic -- univariate -- Boltzmann sampling.
    Words are drawn from a weighted distribution -- \Olive{rejecting those whose frequencies }
   are too distant from the targeted one -- until an acceptable one is found and returned.
    This gives the following rejection sampler $\Gamma_3\mathcal{A}{(x,\Weights;n,\boldsymbol{m},\varepsilon,\sigma)}$ for a language $\mathcal{A}$ where $x$ is real,
    $\Weights$ a real $k$-vector, $\boldsymbol{m}$ a map from $\mathbb{N}$ to $\mathbb{R}^k$, and $\varepsilon$ the tolerance:\\
\begin{algorithm}[H]
\begin{framed}
\caption{$\Gamma_3\mathcal{A}{(x,\Weights;n,\boldsymbol{m},\varepsilon,\sigma)}$ }

\KwIn{The parameters $x,\Weights ,n,\boldsymbol{m},\varepsilon,\sigma$}

\KwOut{An object of $\mathcal{A}$ of size $s$ in $I(n, \varepsilon) $} and for every parameter $\WF_i$, the number of occurrences of $Z_i$ is in $I(m_i(s), \varepsilon,\sigma) := [m_i(s)-m_i(s)^\sigma\varepsilon, m_i(s)+m_i(s)^\sigma\varepsilon]$\\
\Repeat{ $\forall i, |\gamma|_i \in I(m_i(s), \varepsilon,\sigma)$ }{$\gamma := \Gamma_2\mathcal{A}{(x,\Weights;n,\varepsilon)}$}
\Return($\gamma$)
\end{framed}
\end{algorithm}

In many important classes of combinatorial structures, the composition of a random object is concentrated around its mean. It follows that
a rejection-based generation can succeed after few attempts, provided that the expected composition matches the targeted one.
Our main result is that, for any irreducible and simple context-free language, a suitably parameterized multidimensional rejection sampler
generates a word of targeted composition after $\BigO{n^{k/2}}$ attempts.
Moreover, allowing a $n^\beta$ ($\beta>1/2$) tolerance on the number of occurrences of each letters yields a sampler that succeeds in expected number
of attempts asymptotically constant.

Now, let us denote by $U_n(\Weights_0)$ the $k$-multivariate random variable which follows the probability $$\mathbb{P}(U_n(\Weights_0)=\boldsymbol{a})=\frac{\Weights_0^{\boldsymbol{a}}\cdot[z^n \Weights^{\boldsymbol{a}}]C_{\Weights}(z)}{[z^n]C_{\Weights_0}(z)},$$ i.e. the distribution of the parameters for objects of size $n$.
Moreover, let us denote by $\boldsymbol{\mu}(n,\Weights_0)$ the mean-vector of $U_n(\Weights_0)$ and by $\boldsymbol{V}(n,\Weights_0)$ its variance-covariance matrix.
  If we do not have any strict correlation between the parameters, the matrix $\boldsymbol{V}(n,\Weights_0)$ is positive definite (and so, invertible). We can then \YCorr{associate a norm to each composition vector $\boldsymbol{u}$ through} $||\boldsymbol{u}||_{\boldsymbol{V}^{-1}}:=\sqrt{\boldsymbol{u}^T\boldsymbol{V}(n,\Weights_0)^{-1}\boldsymbol{u}}.$ Now, let $\boldsymbol{V}$ be a positive definite matrix, we denote by $\kappa\left(\boldsymbol{V}\right):=\inf\limits_{||\boldsymbol{u}||_\infty=1}\{||\boldsymbol{u}||_{\boldsymbol{V}} \}$, the infinum distance\footnote{Recall that the infinity norm is defined as $||\boldsymbol{u}||_\infty=\max{(|u_1|,\cdots,|u_k|)}$} from the unit sphere to the center of the  Banach space.
\begin{definition}
The \emph{$\sigma$-concentrated condition} is defined as :\\
 $$\limsup_{n\to \infty} \left(||\boldsymbol{\mu}(n,\Weights)||_\infty\right)^\sigma\cdot\kappa\left(\boldsymbol{V}(n,\Weights)^{-1}\right)=c>\sqrt{k}/\varepsilon.$$
\end{definition}

\begin{theorem}[Approximate composition]\label{theo3}
Let $x_n$ and $\Weights_{\boldsymbol{a}}$ be the solution of  $\mathbb{E}_{x,\Weights}(N)=n$ and $\mathbb{E}_{x,\Weights}(N_i)=a_i$. The map  $\boldsymbol{m}$ is defined as the \Olive{$\boldsymbol{m}: s\mapsto \mathbb{E}_{s,\Weights_{\boldsymbol{a}}}(N_i)$}  and assume that
the $\sigma$-concentrated condition holds for some $\sigma\leq 1$.
Then the expected number of trials (of  $\Gamma_2\mathcal{A}{(x_n,\Weights;n,\varepsilon)}$) of the rejection sampler $\Gamma_3\mathcal{C}{(x_n,\Weights_{\boldsymbol{a}};n,\boldsymbol{m},\varepsilon,\sigma)}$ is upper-bounded by  $$\sup_{s \in I(n, \varepsilon) }\frac{\left(\varepsilon \cdot\kappa\left(\boldsymbol{V}(s,\Weights_{\boldsymbol{a}})^{-1}\right)\cdot||\boldsymbol{\mu}(s,\Weights_{\boldsymbol{a}})||^\sigma_\infty\right)^2}{\left(\varepsilon\cdot \kappa\left(\boldsymbol{V}(s,\Weights_{\boldsymbol{a}})^{-1}\right)\cdot||\boldsymbol{\mu}(s,\Weights_{\boldsymbol{a}})||^\sigma_\infty\right)^2-k}$$ which tends to a constant as $n\rightarrow \infty$.\\
\end{theorem}
%
\begin{theorem}[Exact composition]\label{theo4}
Assume that $(U_n(\Weights_{\boldsymbol{a}}))$ 
 admits a multidimensional Gaussian law with mean $\boldsymbol{\mu}$ and variance-covariance matrix $\boldsymbol{V}$ proportional to $f(n)$ as limiting distribution when $n$ tends to the infinity,
then the exact-composition rejection sampler $\Gamma_3\mathcal{C}{(x_n,\Weights_{\boldsymbol{a}};n,\boldsymbol{m},0,1)}$ succeeds after an expected number of trials equal to $(2\pi)^{k/2}(\det(\boldsymbol{V}))^{1/2}=\BBigO{f(n)^{k/2}}$.
\end{theorem}
 \begin{proof} Just notice that the probability to draw an exact composition corresponds to take $\boldsymbol{u}=\boldsymbol{\mu}$ in the asymptotic estimate $$p(\boldsymbol{u}) = \frac{1}{ (2\pi)^{k/2}\left(\det(\boldsymbol{V})\right)^{1/2} }
             \exp\Big( {-\tfrac{1}{2}}(\boldsymbol{u}-\boldsymbol{\mu})^t \; \boldsymbol{V}^{-1}(\boldsymbol{u}-\boldsymbol{\mu})+o(1) \Big).$$
             Consequently the expected number of attempts is  $(2\pi)^{k/2}\det(\boldsymbol{V})^{1/2}=\BigO{f(n)^{k/2}}.$
 \end{proof}

  \subsection{Rational languages: Bender-Richmond-Williamson theorem}
The Bender-Richmond-Williamson theorem~\cite[Theorem~1]{Bender1983b} defines sufficient conditions such that the limiting
distribution of a rational language $\CS{R}$ is a multidimensional Normal distribution.
Let us remind that a rational language is \Def{irreducible} if its minimal automaton $\mathcal{A}$ is strongly-connected, and
\Def{aperiodic} -- if the cycle lengths in $\mathcal{A}$ have greatest common divisor equal to $1$.
Additionally the \Def{periodicity parameter lattice} $\Lambda$, defined in~\cite{Bender1983b} (Definition 2) is required to be full dimensional
to avoid trivial correlations in the occurrences of letters.

\begin{theorem}\label{th:rational}
Let $\CS{R}_{\Weights}$ be a weighted rational language whose minimal automaton is irreducible and aperiodic, and
$x_n$  be the root in $[0, \rho_{\Weights})$ of  $\mathbb{E}_{x,\Weights}(N)=n$.
Assume that the periodicity parameter lattice $\Lambda$  is full dimensional; Then:
\begin{itemize}
\item $\forall \sigma>1/2$, the approximate-composition sampler $\Gamma_3 \CS{R}(x_n,\Weights; n,\varepsilon,\sigma)$ succeeds after $\BigO{1}$ trials
\item For $\sigma=1/2$, $\exists \varepsilon_0$ such that $\forall  \varepsilon>\varepsilon_0$ $\Gamma_3 \CS{R}(x_n,\Weights; n,\varepsilon,\sigma)$ succeeds after $\BigO{1}$ trials
\item The exact-composition rejection sampler $\Gamma_3 \CS{R}(x_n,\Weights; n,0,1)$ succeeds after $\BigO{n^{k/2}}$ trials.
\end{itemize}
\end{theorem}
\begin{proof}
From the system of language equations $\boldsymbol {\mathcal{L}}={\boldsymbol M}\cdot \boldsymbol{\mathcal{L}}+\boldsymbol{\mathcal{E}}$,
we directly obtain the system $\boldsymbol {{L}}=z{\boldsymbol M}\cdot \boldsymbol{{L}}+\boldsymbol{{E}}$ for the generating function.
In this case the  Perron-Frobenius theorem ensures that the dominating pole of every $L_i$ in $\boldsymbol{L}$ is \YCorr{the smallest positive real root}
of $\det(\mathbb{I}-z\cdot{\boldsymbol M})=0$ and that this pole is simple. Now, assume that the  \emph{periodicity parameter lattice}
$\Lambda$ defined in~\cite{Bender1983b} (Definition 2) is full dimensional.
Assume also that we have a compact set $\Pi_1$ for the parameters in which the singular exponent is constant and equal to $1$. Then from the Bender-Richmond-Williamson theorem (see~\cite{Bender1983b},
Theorem 1 and~\cite{Bender1983a}), it follows that for any fixed parameter in the compact set $\Pi_1$, the limiting distribution of the parameters is a multidimensional Gaussian distribution with
mean and variance-covariance matrix proportional to $n$. Consequently, Theorem~\ref{theo3} applies for $\sigma>1/2$, Theorem~\ref{theo4} applies with $f(n) = n$, and the result follows.
\end{proof}

Let us discuss the prerequisites of Theorem~\ref{th:rational}. If the matrix $\boldsymbol{M}$ is not aperiodic, there exists a power $d$ such that $\boldsymbol{M}^d$ is aperiodic. So, we can always reduce the problem to a list of $d$ aperiodic ones, and Theorem~\ref{th:rational} applies under the same assumptions (full dimensional periodicity parameter lattice and compact set with constant singular exponents).
The \emph{irreducibility} requirement may be lifted when one of the strongly connected components dominates asymptotically, i.e. when the associated schema only involves subcritical and supercritical compositions \cite[Theorem~IX.2]{Flajolet2009}. However the case of a competition between different components in a non irreducible automaton is much more challenging and requires serious developments that cannot be included in this short paper. Finally we point out that, with minor modifications, similar results could be obtained for more general transfer matrix models.

\subsection{Context-free languages: Drmota theorem}
A theorem of~\cite{Drmota1997} gives very similar sufficient conditions for the limiting multivariate distribution to satisfy the conditions of Theorem~\ref{theo4}.
Namely, the irreducibility condition needs \YCorr{being} fulfilled by the \Def{dependency graph} of the grammar -- the directed graph on non-terminals whose
edges connect left hand sides of rules to their associated right-hand sides. The lattice and aperiodicity properties are replaced by the very similar concept of
\Def{simple type} grammar, requiring the existence of a \emph{positive} $k+1$ dimensional cone centered on $\boldsymbol{0}$ in the space of coefficients.

\begin{theorem}\label{th:contextfree}
Let $\CS{C}_{\Weights}$ be a weighted context-free language \YCorr{generated from a grammar $\mathcal{G}$
of simple-type \cite[Theorem~1]{Drmota1997} and whose dependency graph is strongly connected.}
Then the complexities summarized in Theorem~\ref{th:rational} also hold for $\CS{C}_{\Weights}$.
\end{theorem}

Again, the strong-connectedness requirement could be relaxed for disconnected grammars whose behavior is dominated by that of a single connected component. A formal characterization of such grammars \Olive{can be interpreted in} the theory of (sub/super)-critical compositions \cite[Theorem~IX.2]{Flajolet2009}.

\section{Sampling perfect Tetris tesselations}
  \newcommand{\TG}{\mathcal{T}}
  \newcommand{\B}{\mathcal{B}}
  \newcommand{\Scale}{0.2}
  \newcommand{\SBl}{\includegraphics[scale=\Scale]{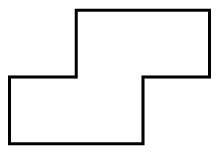}}
  \newcommand{\ZBl}{\includegraphics[scale=\Scale]{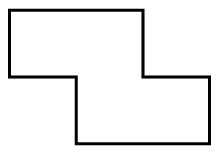}}
  \newcommand{\LBl}{\includegraphics[scale=\Scale]{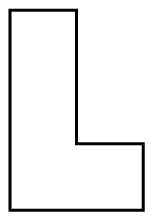}}
  \newcommand{\JBl}{\includegraphics[scale=\Scale]{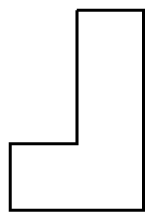}}
  \newcommand{\OBl}{\includegraphics[scale=\Scale]{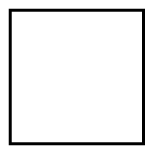}}
  \newcommand{\IBl}{\includegraphics[scale=\Scale]{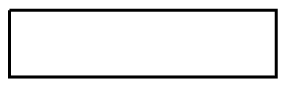}}
  \newcommand{\TBl}{\includegraphics[scale=\Scale]{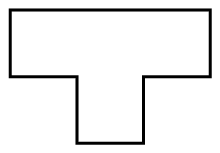}}
  \newcommand{\SSSyst}{\includegraphics[scale=\Scale]{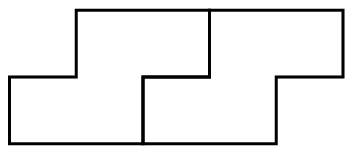}}
  \newcommand{\ZZSyst}{\includegraphics[scale=\Scale]{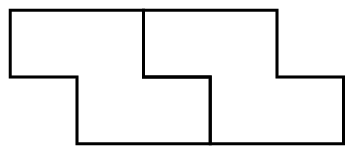}}

  \newcommand{\Bound}{\mathcal{B}}
  \newcommand{\Pieces}{\mathcal{P}}
  \newcommand{\Piece}{p}

  In this short illustration, we address the generation of Tetris tessellations, i.e. tessellations  using tetraminoes of a board having prescribed width $w$.
  The Tetris game consists in placing falling tetraminoes (or {\bf pieces}) $\Pieces$ in a $w\times h$ board.
  The goal of the player is to create hole-free horizontal lines which are then eliminated, and the game goes on until \YCorr{some piece stacks
  past the board ceiling}.
  Most implementations of Tetris use the so-called \emph{bag strategy}, which consists in giving the player sequences of permutations of the 7 types of tetraminoes,
  therefore inducing a \Olive{uniform} composition in each tetramino type. A rational specification (Built by Algorithm~\ref{algo:statespace}) exists for Tetris
  tessellations of any fixed width, but the additional constraint on composition provably throws the associated language out of the context-free class.
  Therefore, we choose to model the generation of uniformly distributed Tetris tessellations as a multivariate generation \YCorr{within} a rational language.
  Such tessellations could in turn be used as a basic construct to build hard instances for the offline version of the algorithmic
  Tetris problems studied in~\cite{Breukelaar2004} and \cite{Hoogeboom2004}.

  \subsection{Building the automaton of Tetris tesselations}
\begin{algorithm}[t]
\begin{framed}
\begin{tabular}{cc}
\begin{minipage}{.46\textwidth}
\KwIn{Board width $w$ and flat boundary $\Bound_w$}
\KwOut{$Q$ the states set and $\sigma$ the transition \\ \quad\quad  function of $\mathcal{A}_w=\left(\Pieces,Q,\Bound_w,\{\Bound_w\},\sigma \right)$}
\Begin
{
  $(Q,\sigma)\leftarrow(\Bound_w,\varnothing)$\\
  $S \leftarrow \{\Bound_w\}$\\
  \While{$S \neq \varnothing$}
  {
    $S \Rightarrow_{\text{pop}} \Bound$\;
    \For{$p \in \mathcal{P}_{\Bound}$}
    {
      $\Bound' \leftarrow \Bound- p$\;
      \If{$\Bound'\notin Q$}
      {
        $Q \leftarrow Q \cup \{\Bound'\}$\;
        $S \Leftarrow_{\text{push}} \Bound'$\;
      }
      $\sigma \leftarrow \sigma \cup \{(\Bound,p,\Bound')\} $\;
    }
  }
  \Return{$(Q,\sigma)$}
}
\end{minipage} & \begin{tabular}{c|c|c}
        Width $w$ & \#States in $\mathcal{A}_w$ & \#States minimal\\ \hline
        2 & 4 & \YCorr{4}\\
        3 & 55 & \YCorr{55}\\
        4 & 80 & 78\\
        5 & 1686 & 1646\\
        6 & 4247 & 4130\\
        7 & 41389 & 40099\\
        8 & 49206 & 47564\\
        9 & 919832 & --
    \end{tabular}
\end{tabular}
\caption{Constructing the automaton $\mathcal{A}_w$ for tessellations of width $w$. Right: Growth of the number of states for increasing values of $w$.}\label{algo:statespace}
\end{framed}
\end{algorithm}

    First let us find an unambiguous decomposition of Tetris tessellations.
    The idea is to focus on the state of the upper band of the tessellation \YCorr{of height $4$}, or \Def{boundary} of a partial tessellation.
    In particular for (complete) Tetris tessellations the upper band is completely filled and the associated boundary is \Def{flat}.
    One can investigate the different ways to get to a given boundary $\Bound$ by simulating the \Def{removal} from $\Bound$ of a
    piece $\Piece$\YCorr{, completing the boundary after each removal so that the highest non-empty position stays on the top row.}
    Without further restriction on the position of removal, such a decomposition would be \emph{ambiguous} and give rise to an infinite number
    of different boundaries. Consequently, we enforce a canonical order on the removal of pieces
    by restricting it to a set of (possibly rotated) pieces $\mathcal{P}_{\Bound}$ positioned such that: \YCorr{a) the upper-rightmost position of the piece
    matches that of the boundary and  b) the piece is entirely contained in the boundary.
    We refer to the induced decomposition as the \Def{disassembly decomposition}}.

    \begin{proposition}
      The disassembly decomposition generates sequences of removals from and to flat boundaries that are in bijection with Tetris tesselations.
    \end{proposition}
    \begin{proof}[Sketch]
    Let us discuss briefly the correctness of this decomposition, or equivalent that the sequences of $k$ removals leading from a \Def{flat boundary}
    $\Bound_w$ to itself are in bijection with the tessellations of width $w$. First let us notice that the decomposition is unambiguous,
    since all the local removals share at least one position (the upper-rightmost of the boundary)
    and are therefore strongly ordered. Furthermore \Olive{the \YCorr{decomposition} is also provably complete by induction on the number $n$ of piece}, since any
    tessellation has a upper-rightmost position which, upon removal, gives another tessellation of smaller size, and completeness of
    the decomposition propagates from tessellations of size $n$ to size $n+1$. Finally, it gives rise to a finite number of states since
    the difference between the highest and lowest point in any reached boundary does not exceed the maximal height of a piece.
    \end{proof}

    The finiteness
    of the state space suggests Algorithm~\ref{algo:statespace} that builds the automaton $\mathcal{A}_w$, generating tessellations of width $w$. Notice that the resulting automaton in not necessarily co-accessible, since the removal of some piece can create boundaries that cannot be completed
    into a flat one through any sequence of removal.
    Consequently, we added in our implementation a test of connectedness that discards any boundary having a (dis)connected component involving a number
    of blocks that is not a multiple of 4, as such boundaries clearly cannot reach a flat state again.
    Running a minimization algorithm of the resulting automata confirms the expected explosion in the number of states (See Algorithm~\ref{algo:statespace}) required for increasing values of $w$.

  \subsection{Random generation}
    First we point out that the automaton has matching initial and final states, so the strong connectedness is obviously ensured
    and our theorems regarding the complexity of \YCorr{our generators} apply. One can then translate \YCorr{the automaton transitions} into a system of functional equations
    involving the (rational) generating functions associated with each states. Solving the system gives the generating functions, from which one can extract many
    informations.

    For instance, fixing the width $w=6$ and a number  $n=105$ of pieces, one obtains a number $h_{6,105} = 3.10^{71}$ of potential tessellations, and extracting
    coefficients of suitable derivatives yields:

    {\centering
    \begin{tabular}{c@{\hspace{1em}}c@{\hspace{1em}}c@{\hspace{1em}}c@{\hspace{1em}}c@{\hspace{1em}}c@{\hspace{1em}}c@{\hspace{1em}}c}
    Piece &            \ZBl&   \OBl&   \LBl&  \JBl&   \IBl& \SBl & \TBl\\
    Frequency ($\%$) & 7.90& 10.55 & 20.42 & 20.42& 17.00 &7.90& 15.81\\
    \end{tabular}\\} \medskip
   Consequently, the average composition of a Tetris tessellation is incompatible with the \emph{bag strategy}, which \YCorr{induces
   uniformly distributed pieces.} \YCorr{One can then use the results of Section~\ref{tuning} to compute a set of weights
   that ensures $1/7$-th proportions in each type of pieces.}

    {\centering
    \YCorr{\begin{tabular}{c@{\hspace{1em}}c@{\hspace{1em}}c@{\hspace{1em}}c@{\hspace{1em}}c@{\hspace{1em}}c@{\hspace{1em}}c@{\hspace{1em}}c}
    Piece &            \ZBl&   \OBl&   \LBl&  \JBl&   \IBl& \SBl & \TBl\\
    Weight & 0.93&0.84 &0.38 &0.38 &0.46&0.93&0.42\\
    Frequency ($\%$) & 14.3& 14.1& 14.2& 14.2& 14.2& 14.3& 14.5
    \end{tabular}\\}}
    \medskip

    A weight random generation for the $w=6$ and $n=105$, coupled with a rejection that allows the numbers of any piece to
    be equal to $15\pm 1$, gives the instances drawn in Figure~\ref{fig:example}.

  \subsection{From random Tetris tessellations to Tetris instances}
  \begin{figure}[t]
    \centering \includegraphics[width=32em]{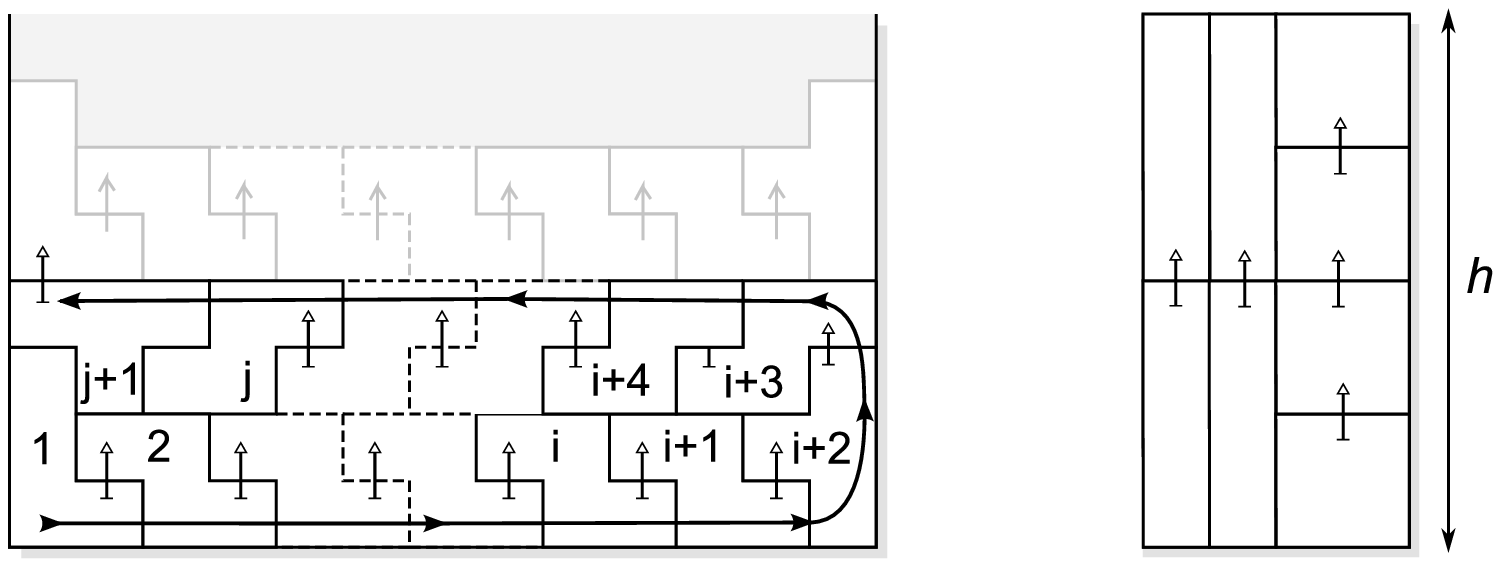}
    \caption{{\bf Left:} Tetris tessellations associated with a unique instance.
    Only the most relevant dependency points are displayed here (arrows) and pieces are labelled with their rank in the only compatible instance. Duplicating the gadget preserves the uniqueness of the associated instance while allowing for the generation of tessellations of arbitrarily large dimensions. {\bf Right:} Tesselation realized by $\binom{h}{h/2}\in \Theta(2^n/\sqrt{n})$ different instances.}\label{fig:unique}
  \end{figure}
   \begin{proposition}
    For any Tetris tesselation $\TG$, there exists an instance (sequence of pieces) such that $\TG$ can be obtained.
  \end{proposition}
  \begin{proof}
    Let us assume that $\TG$ is a tessellation of a $w\times n$ rectangle using tetraminoes,
    and let us call \emph{dependency point} any contact between the southward face of a piece $\B_1$ and the northward
    face of a piece $\B_2$. \YCorr{Such points induce dependencies $\B_1 \to \B_2$, which are the arcs of a \emph{dependency graph} $D=(\TG,E)$.}
    Additionally, each edge is labelled with the coordinate of its associated dependency point.

    It can be shown that $D$ is acyclic, by pointing out that any path along $D$ is labelled with coordinates that are either
    increasing on the $y$-axis or monotonic on the $x$-axis. Let us start by noticing that, aside from the $\ZBl$ and
    $\SBl$ pieces, all types of pieces exhibit northward faces that are strictly higher than their southward ones.
    Furthermore,  \YCorr{any assembly of distinct pieces exposes northward faces that are at greater $y$-coordinates than their
    dependency point}, inducing an increase of $y$-coordinate in the path.
    Consequently,  there only exists two configurations of dependent pieces $A\to B$, namely $\SSSyst$ and $\ZZSyst$, such that
    $B$ exposes a southward face at the same height as their dependency point. The only way for a path in $D$ not to increase
    in $y$-coordinate is then to feature a sequence of $\SBl$ (resp. $\ZBl$) pieces, inducing a monotonic behavior which
    proves our claim, and the acyclic nature of $D$ follows.
    \YCorr{Finally, the acyclicity of $D$ implies the existence of a sequence of pieces realizing $\TG$, since it is always possible to
    removing a piece.}
  \end{proof}
  \begin{figure}[t]\newcommand{\WT}{2.5em}
    \centering \includegraphics[width=\WT]{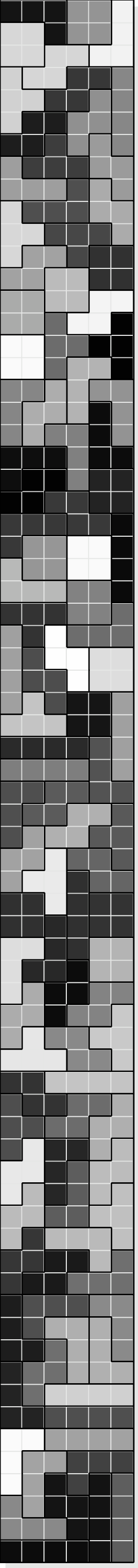}
    \includegraphics[width=\WT]{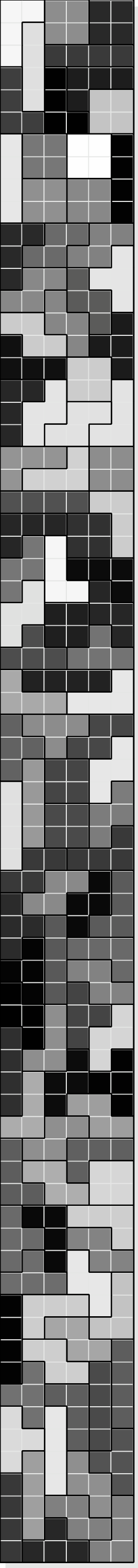}
    \includegraphics[width=\WT]{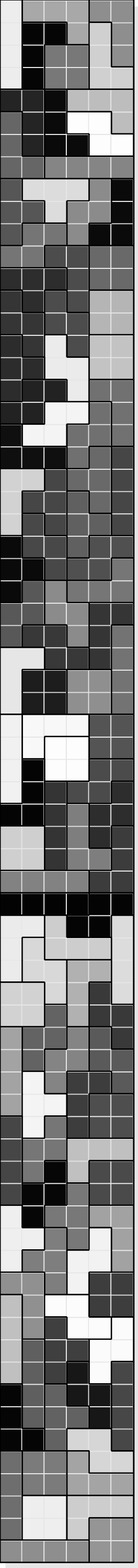}
    \includegraphics[width=\WT]{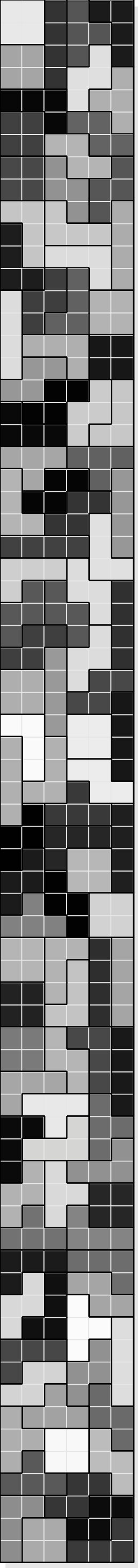}
    \includegraphics[width=\WT]{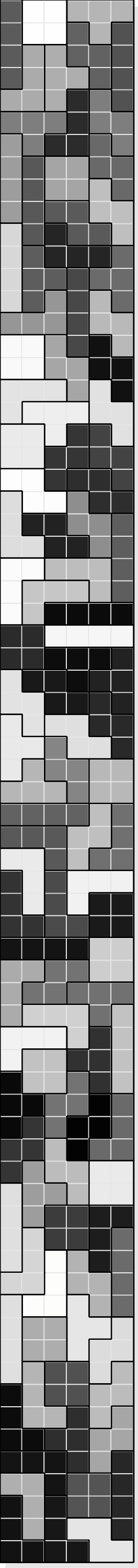}
    \includegraphics[width=\WT]{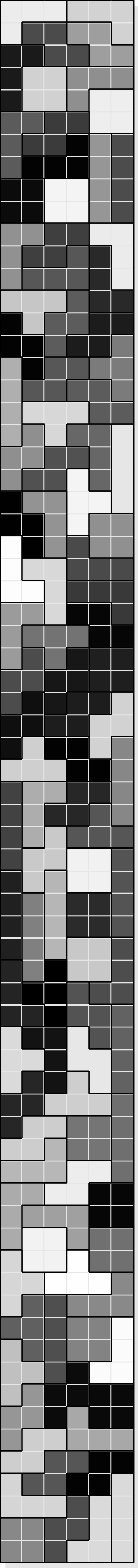}
    \includegraphics[width=\WT]{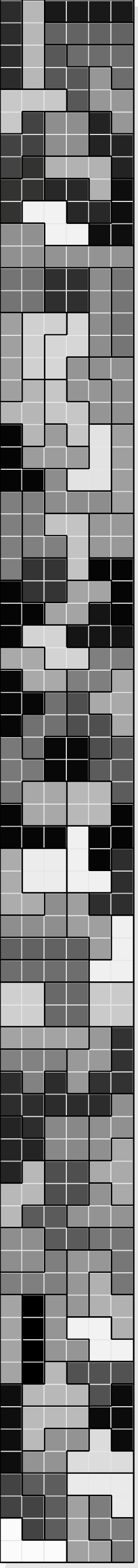}
    \includegraphics[width=\WT]{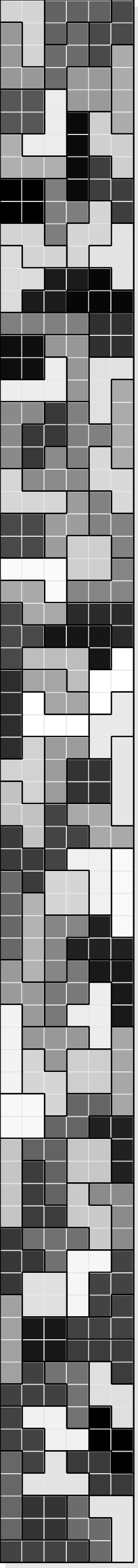}
    \includegraphics[width=\WT]{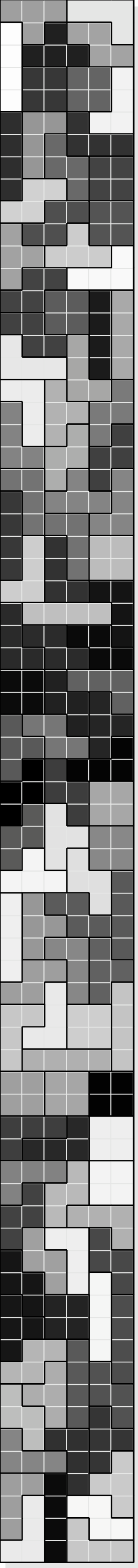}
    \includegraphics[width=\WT]{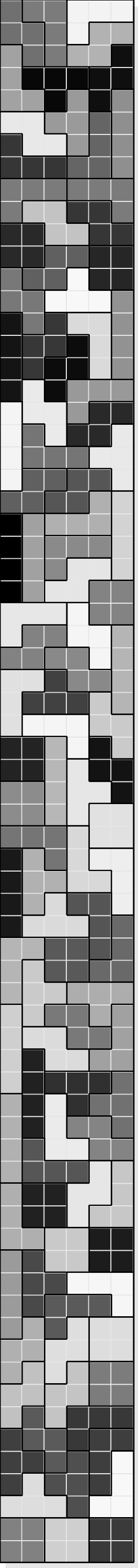}
    \includegraphics[width=\WT]{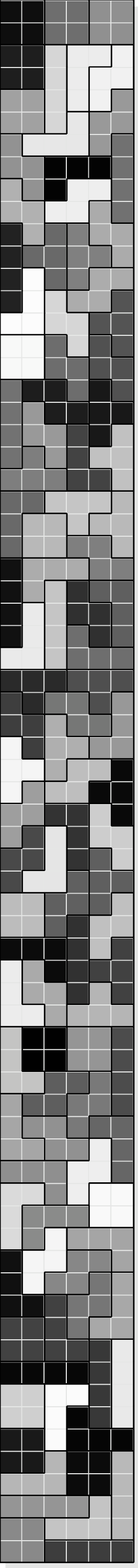}
    \includegraphics[width=\WT]{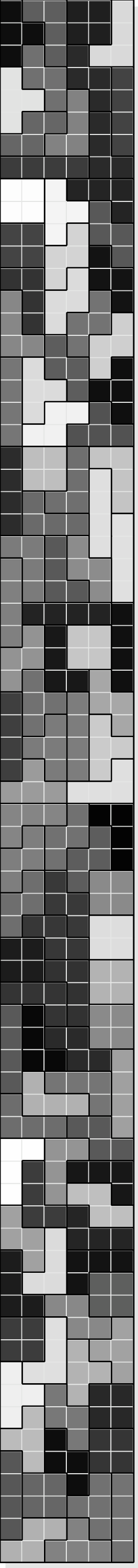}
    \includegraphics[width=\WT]{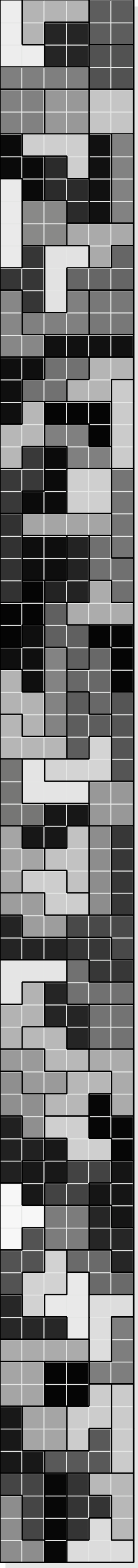}
    \includegraphics[width=\WT]{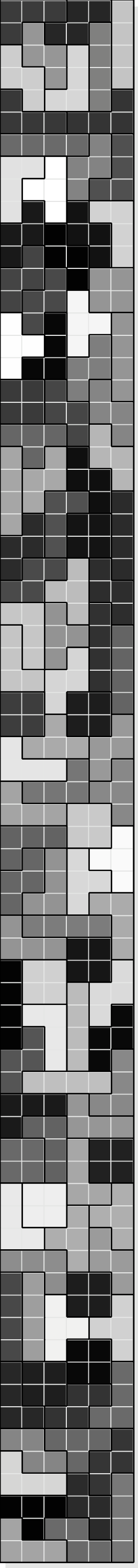}
    \includegraphics[width=\WT]{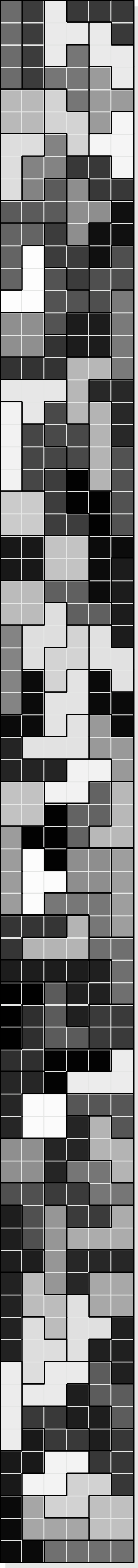}
    \caption{Fifteen Tetris tesselations of width 6 having uniform composition (+/- 1) in the different pieces.}\label{fig:example}
  \end{figure}

  \YCorr{Let us discuss the limitations induced by Tetris tesselation as a model for Tetris instances. First it can be remarked that Tetris tesselations
  do not capture every possible Tetris game ending with an empty board, as one may temporarily leave \emph{holes} which amount
  to disconnecting pieces in the tesselation representation.
  Secondly there generally exists different free pieces to choose from while rebuilding a tesselation, and therefore different instances can lead to a given tesselation.
  Furthermore the number of instances highly depends on the actual tesselation (from one to an exponential in $n$, as illustrated in Figure~\ref{fig:unique}),
  Consequently, using the DAGs associated with Tetris histories to draw instances for the offline version of Tetris algorithmic problems, studied in~\cite{Breukelaar2004},
  would favor exponentially certain instances over others, and the uniform random generation of instances ensuring feasibility of a perfect Tetris game
  remains a challenging problem.}

\section{Conclusion}

In this paper, we adapted and applied a general methodology for the multivariate random generation of combinatorial objects.
Under explicit and natural conditions, random generators having complexity in $\mathcal{O}(n^{2+k/2})$ were derived for the exact size
and composition generation, outperforming best known algorithms (in $\mathcal{O}(n^k)$ and $\mathcal{O}(n^{2k})$ respectively for rational
and context-free languages) for this problem. Furthermore, provided a small (linear) tolerance is allowed on the size of generated objects,
and a $\Omega(\sqrt{n})$ one is allowed in the other dimensions, our generators generate objects in linear expected time. We applied these
principles to the generation of perfect Tetris tessellations with uniform statistic in tetraminoes and discussed the generation of Tetris
games from this model.

This paper is the first step toward a general analysis of the multi-parameters Boltzmann sampling. Compared to its alternative using the recursive
method, the resulting method is not only theoretically faster, but also only requires $\BigO{n}$ storage and its time complexity seems less
affected by larger specifications.
Nevertheless, many questions are left open, for instance with respect to the nature of the dependency between the
weights and \emph{reasonable} frequencies, which would allow us to address the complexities of Phase 2 in a much more general setting.
Furthermore the success of our programme critically depends on the existence of suitable weights, which is not guaranteed, e.g. when the targeted
distribution is incompatible with some dependencies induce by the grammar. A future direction of this work might investigate non-trivial, sufficient -- yet tight --
conditions such that the targeted composition can be achieved on the average.

Since multivariate Boltzmann samplers can be obtained in any situation where the distribution is well-concentrated, one may envision extensions to other classes, including
constrained trees, permutations with a fixed number of cycles, functional graphs with a controlled number of components\ldots
A first extension may consider simple Polya operators and extend some of the multivariate theorems established in the present work.
The requirement of strong-connectedness (or irreducibility) could be questioned or categorized using (sub/super)-critical
compositions. Another direction is the use of Hwang's  Quasi-powers theorem, giving rise to low variance distributions, for a general treatment
of the bivariate case.

\section*{Acknowledgements}
\YCorr{The authors wish to express their gratitude toward Pierre Nicodeme for his thorough inspection of a preliminary version of the manuscript.
This work was supported by the ANR-GAMMA 07-2$\_$195422 grant of the French \emph{Agence Nationale de la Recherche}.}

\bibliographystyle{amsplain}
\bibliography{biblio}

\end{document}